\documentclass[a4paper,fleqn]{cas-sc}

\usepackage[authoryear,longnamesfirst]{natbib}
\usepackage{tikz}
\usepackage{caption}
\usepackage{subcaption}
\def\tsc#1{\csdef{#1}{\textsc{\lowercase{#1}}\xspace}}
\tsc{WGM}
\tsc{QE}
\tsc{EP}
\tsc{PMS}
\tsc{BEC}
\tsc{DE}

\begin{document}
\let\WriteBookmarks\relax
\def\floatpagepagefraction{1}
\def\textpagefraction{.001}

\shorttitle{Decreasing Wages in Gig Economy}

\shortauthors{Pravesh Koirala et~al.}

\title [mode = title]{Decreasing Wages in Gig Economy: A Game Theoretic Explanation Using Mathematical Program Networks}                      

%
\author[1]{Pravesh Koirala}[type=editor,
                        auid=000,bioid=1]

\cormark[1]

\ead{pravesh.koirala@vanderbilt.edu}

\ead[url]{praveshkoirala.com}


\affiliation[1]{organization={Vanderbilt University},
    addressline={}, 
    city={Nashville},
    postcode={TN}, 
    country={United States}}

\author[1]{Forrest Laine}

\cortext[cor1]{Corresponding author}
\cortext[cor2]{Principal corresponding author}

\begin{abstract}
Gig economy consists of two market groups connected via an intermediary. Popular examples are rideshares where passengers and drivers are mediated via platforms such as Uber and Lyft. In a duopoly market, the platforms must compete to attract not only the passengers by providing a lower rate but also the drivers by providing better wages. While this should indicate better driver payout, as platforms compete to attract the driver pool, real world statistics does not indicate such. This goes completely against the intuition that the worker side of a gig economy, given their importance, should always earn better. We attempt to answer the low wages of drivers in the gig economy by modeling the ridesharing game between duopoly platforms, drivers, and passengers using Mathematical Program Networks. Our model is parsimonious, expressive, models the same-side and cross-side externalities of the economy, and has interpretations under both single-homing and multi-homing regimes. We derive the conditions for the existence of a profitable duopoly and show that it can only happen if the platforms collude together to pay the bare minimum to the drivers. This not only answers why drivers are paid less but also provides strong managerial insights to any interested policy maker.
\end{abstract}




\maketitle

\newtheorem{definition}{Definition}
\newtheorem{remark}{Remark}
\newtheorem{assumption}{Assumption}
\newtheorem{theorem}{Theorem}
\newtheorem{lemma}{Lemma}
\newcommand{\real}{\mathbb{R}}
\newcommand{\opt}[1]{{#1}^*}
\newproof{proof}{Proof}

\section{Introduction}
Gig economies\footnote{Also known as two-sided economy, sharing economy, platform economy, online economy, on-demand economy, double-sided network economy etc. in the popular literature.} that serve as an intermediary between a supply and a demand market have seen an unprecedented rise following the ubiquity of digital technologies. Rideshare, a popular example of the gig economy and the target of this study, alone has been reported to have a global market size of 96.9 Billion USD with an estimated 45.2 Billion USD market cap in the U.S. alone\footnote{https://www.statista.com/topics/4610/ridesharing-services-in-the-us}. In 2023, reportedly 23 million Americans engaged in earning activities via online platforms\footnote{https://www.cnn.com/2023/07/24/economy/gig-workers-economy-impact-explained/index.html}. Similarly, an estimated 26\% of the U.S. workforce are currently engaged in the gig economy and Forbes estimates that this number may grow up to 50\% by 2027\footnote{https://www.forbes.com/sites/forbestechcouncil/2023/08/08/how-to-better-support-gig-workers-through-payment-modernization/}.
Considering these ballooning numbers, it seems only reasonable that policy-makers be acutely invested in the impact gig economy has on the overall global and national economy and the collective social welfare of its participants. In particular, the suppliers (drivers) of these platforms need a focused interest from law-makers. Unlike the consumers, who merely treat it as a commodity, drivers are increasingly (and sometimes exclusively) dependent upon the platforms' wage for sustenance \citep{chen_bonus_2022}. And while laws are progressively being proposed for drivers' welfare protection\footnote{https://www.msn.com/en-us/news/us/uber-and-lyft-drivers-will-get-expanded-death-benefits-thanks-to-new-washington-state-law/ar-BB1j41r9}$^,$\footnote{https://www.startribune.com/new-battle-lines-in-minneapolis-debate-over-uber-and-lyft-driver-pay-same-as-the-old-lines/600346563/} our study shows that it may not still be enough for ensuring that they are not entirely exploited by the platforms.

To argue about deficiencies in the existing policies surrounding the gig economy, we need a strong theoretical basis. And while there have been previous attempts for a strategic modelling of such economies (see section \ref{sec:litreview}), it's generally known to be a difficult undertaking. To begin with, online economies have myriads of strategic interactions among multiple constituent actors. In ridesharing for instance, there are platforms like Uber, Lyft, Didi, etc. who decide what to charge and how much to pay to, respectively, their riders and drivers, who, in turn, decide which platform to use and serve.
When choosing a platform to utilize, many of the riders and drivers are typically not tied to a specific platform (i.e. they can use any of the available ones depending upon the costs and the profits) and are, thus, engaged in what's called a \textit{multi-homing} behavior adding further complexity to the model. 

The name two-sided economy comes from the fact that the riders and drivers construct two distinct but connected networks for the competing platforms i.e. the demand network and the supply network. Any model of the gig economy, then, must also account for both the interactions (externalities) within and across these two networks, i.e. the additional utility (or cost) derived by any participant of the network from the presence (or absence) of participants in the same (same-side) or the complementary (cross-side) network. Concretely, a rider is better served when there are more drivers in a platform due to reduced wait-times and vice-versa, which shows that there is a positive cross-side externality between the supply and demand networks. However, a rider (driver) is worse-off when there are more riders (drivers) because of increased wait-times (decreased utilization). In this regard, the ridesharing economy demonstrates a negative same-side externality. Another modeling difficulty is due to the order of interaction. There may be different equilibrium outcomes in the game depending upon who moves/commits to an action first. The choice of the order, however, is not clear. For example, a model where the platforms, drivers, and the riders all act simultaneously may not be indicative of reality where passengers often have the full information (fares, wait-times, etc.) and are the last to decide which platform to use. Whereas, a model in which the platforms commit to the wages and the charge but the drivers and passengers simultaneously decide on a trip may not truthfully emphasize the importance drivers have in the exchange. 

Considering the need to include all aforementioned characteristics, constructing any suitably complete model with parsimony seems to be a challenging task. As we show in this study, mathematical program networks \citep{laine2024mathematical}, an extension of multilevel programming, can simplify much of these aspects and provide a powerful modeling tool for problems of such nature. Mathematical program networks (MPNs) provide a natural modelling mechanism for a sequential game where there are different rational actors choosing their actions in some specified order. They enable a parsimonious depiction of the sequence of a game by virtue of their structure. And while MPNs themselves are relatively recent, their predecessor i.e. multilevel programming has successfully been used in fields such as Economics, Optimal Control, Decision Theory, Logistics etc, to model applications with multiple stakeholders taking strategic actions in a defined order. Some of these applications include fund allocations, supply chain optimization, power systems and security etc \citep{Han2017TrilevelDF, Fard2018ATL, Fard2018HybridOT, Tian2019MultilevelPC, luo2020energy}.

In this paper, we use mathematical program networks to model and analyze a prominent example of the gig economy, i.e. the ridesharing economy, to shed light upon some of its critical interactions. To be explicit, we investigate, inter alia, why the supply side of the economy i.e. the drivers get paid so little and are seemingly exploited.
\begin{figure}
    \centering
    \includegraphics[width=6in]{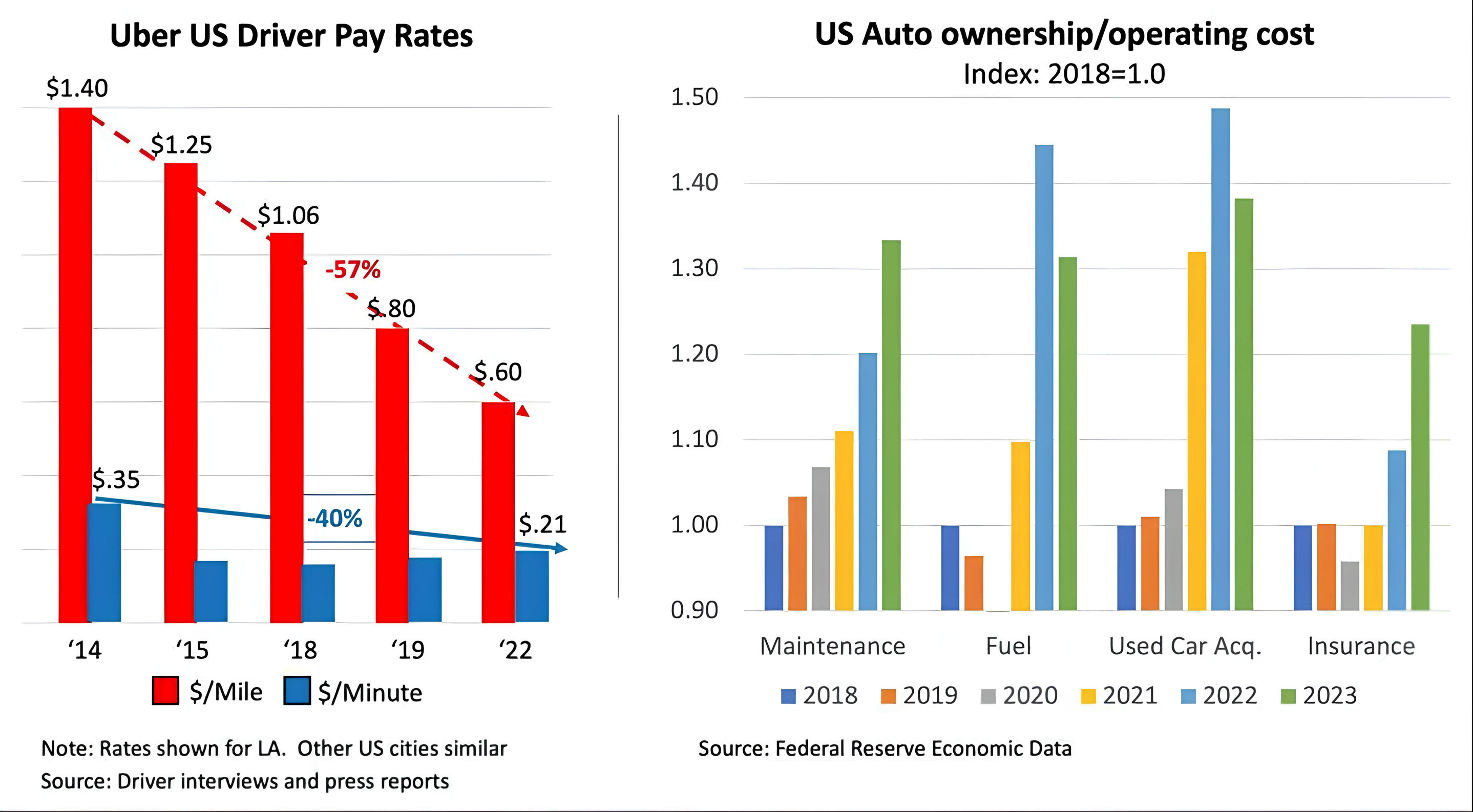}
    \caption[Short Caption]{Uber wages vs cost of automobile operation throughout the years \footnotemark }
    \label{fig:cost-vs-auto}
\end{figure}
\footnotetext{https://www.forbes.com/sites/lensherman/2023/12/15/ubers-ceo-hides-driver-pay-cuts-to-boost-profits}
Fair driver wages have been an important social issue since the inception of ridesharing. In 2018, it was reported that after adjustments, the net effective hourly earnings of an Uber driver was at the 10th percentile of all wage and salary workers' wages\footnote{https://www.epi.org/publication/uber-and-the-labor-market-uber-drivers-compensation-wages-and-the-scale-of-uber-and-the-gig-economy/}. 
Similarly, Forbes estimates that there has been a 57\% deduction in driver's base pay rate from 2014 to 2022\footnote{See footnote 8}. 
The decline is reportedly continuing with average Uber and Uber Eats' drivers seeing a drop in their commissions by as much as 15\% in 2023\footnote{https://www.businessinsider.com/how-much-uber-eats-lyft-doordash-instacart-grubhub-drivers-make-2024-2} and there have been instances of drivers protesting against the low wages with a rising concern that the payout provided by the platforms is not commensurate with what they charge the riders\footnote{https://spectrumnews1.com/ky/louisville/news/2024/03/18/uber-and-lyft-drivers-strike}.
Coupling this with the fact that the workers in such on-demand platforms are not considered salaried and are, thus, not provided with benefits such as health insurance, retirement packages etc, and that cost of automobile operation and maintainance has surged in the recent years\footnote{https://www.minneapolisfed.org/article/2023/despite-easing-inflation-vehicle-repair-costs-soar} there's an increasing concern that the workers are being short-charged. 

From a theoretical perspective, this is a highly perplexing problem. First, it's well-known that double-sided markets of such nature are disproportionately dependent upon the supply-side for existence \citep{chen_bonus_2022}, so it does not make sense for the most crucial aspect of the economy itself to be disadvantaged. Furthermore, in a duopoly competition with positive network externalities, there is known to be a \textit{tipping effect}, i.e. a tendency of a duopoly (or an oligopoly) to collapse into a monopoly \citep{dube2010tipping} when either the supply or demand side of the network is enticed enough to a degree. This, again, should suggest a fierce competition between the platforms to win over the drivers, which should raise their wages significantly. Pre-existing theoretical models also seem to suggest that duopoly should be always better for driver wages \citep{nikzad_thickness_2017}. However, this does not seem to match the observations in the ridesharing economy. For example, Uber and Lyft are inarguably a duopoly. And yet, the wages they provide to their drivers, excluding surge-pricing and other geographical considerations, have remained low and relatively the same and they have, over the years, managed to form a stable duopoly \citep{lian_labor_2022}. Two prominent platforms paying the bare minimum wage to the enabler of their entire enterprise while not \textit{tipping over} strongly hints at some curious interplay that we attempt to explain via our modelling approach.

 Our model suggests that when the drivers and riders engage in a sequential subgame after the prices and wages have been exogenously determined by the platforms, there is indeed a strong inclination for the entire market to tip over into a monopoly which should, theoretically, collapse the equilibrium to one where it's impossible for any platform to make a profit. This part of the result is in agreement with existing literature \citep{bryan_theory_2019, wu_two-sided_2020, chen_bonus_2022, bai_can_2022}. However, we take our results further by showing that any assumption of both platforms making profit in a homogeneous multi-homing market without any service-differentiation trivially implies either a two-sided collusion where both firms collude on prices and wages or a single-sided collusion on wages which must be the bare minimum. To the best of our knowledge, our model is a first complete model of the Ridesharing economy that analyzes possible collusions while taking into account all of the network externalities (same-side / cross-side) for both the demand and supply side and explains the real world problem of drivers having low wages in spite of their importance for the platforms.

The main contributions of our paper is as follows:
\begin{enumerate}
    \item We present a parsimonious model of the ridesharing duopoly competition including all (same-side, cross-side) network externalities as a mathematical program network.
    \item We analyze the equilibrium of our model and present possible collusion modes under the assumption of a non-trivial duopoly.
    \item We present an explanation for low wages of the workers engaged in the economy as the potential result of a single-sided collusion between the platforms.
\end{enumerate}

The rest of the paper is structured as follows: In section \ref{sec:litreview}, we provide a review of the literature surrounding platform economy, ridesharing, mathematical program networks (MPNs), multilevel programming, and collusions. We lay the foundation for defining MPNs in section \ref{sec:MPNs}. We then construct our model in section \ref{sec:model} as a MPN and validate its characteristics by arguing about its ability to capture various network externalities inherent in ridesharing. We proceed to define and analyze the equilibrium solution of the presented model in section \ref{sec:equilibrium} and establish the main results of this paper. Finally, we distill our conclusion in section \ref{sec:conclusion} and provide relevant managerial insights.

\section{Literature Review}
\label{sec:litreview}
\subsection{Platform Economy}
Major early works on platform economy studied newspapers, credit cards, videogames, operating systems etc. in both single-homing and multi-homing contexts \citep{armstrong_competition_2006, rochet_platform_2003, caillaud2001chicken, rochet2006two}. These studies do take into account the cross-side effects between two networks, but do not consider same-side effects like congestion. This is not a valid assumption in many contexts (including rideshares) because it's readily apparent that congestion and resulting delays can shape the market dynamics significantly in an on-demand setting, specially when it comes to a platform competition. There is a rich literature concerning platform economy and we refer readers to some recent surveys \citep{zou_operations_2023, yan_-demand_2022, sanchez-cartas_multisided_2021, jullien_two-sided_2021, trabucchi_landlords_2021, jullien_economics_2021} instead of attempting a comprehensive compilation of our own. Moving onwards, we focus exclusively on works that model platform economy in the context of ridesharing.

\subsection{Ridesharing Duopoly}
Many existing works focus on the analysis of equilibrium market conditions and explore different pricing structures, effects of customer delay sensitivity, spatial pricing, etc. for ridesharing in the presence of a single platform, i.e., a monopoly \citep{taylor_-demand_2016, hu_price_2017, bai_coordinating_2019, bimpikis_spatial_2019, benjaafar_labor_2022}.

Works exploring a duopoly competition in ridesharing contexts are relatively scarce. \citet{nikzad_thickness_2017} study the impact of market thickness on the equilibrium pricing for both monopoly and duopoly using a dynamic steady-state model. They argue that workers' wage and welfare are always high in a duopoly as compared to a monopoly. Similarly, \citet{bernstein_competition_nodate} study a platform competition under congestion and argue that surge pricing increases welfare in comparison to a fixed pricing and that it's worse for drivers and riders alike when they multi-home. \citet{bryan_theory_2019} derive a general theory for multi-homing in a ridesharing game using a hotelling line approach. Their results indicate that in cases where both drivers and riders multi-home, the equilibrium collapses due to a Bertrand competition and no platform generates any profit at all.  \citet{benjaafar_workers_2020} study if platform competition is better for social welfare using a Hotelling line and arrive at the conclusion that under some assumptions, both workers and customers are worse off in a competition. \citet{zhang_inter-platform_2021} study the duopoly competition as a Nash equilibrium and transform it to a Variational Inequality to solve it. They argue that under drivers multi-homing in an unregulated market, the market collapses with no profit for either platforms and social welfare decreases for every participant.
\citet{bai_can_2022} study whether it's possible for two competing platforms to be both profitable at all under different modes of homing, pricing, and market characteristics. They suggest that unless there is some service differentiation or single homing, it's impossible for both platforms to simultaneously earn profit. 

\citet{cohen_competition_2022} study two scenarios where two platforms either compete or cooperate with each other using a general attraction model. They argue that if terms of bargaining are carefully established, cooperation (or monopoly) is better for all parties involved. \citet{siddiq_ride-hailing_nodate} conduct a study in whether competing platforms are better or worse off when provided access to Autonomous Vehicles. \citet{chen_bonus_2022} study the effect of providing bonus to drivers in a platform competition and find that offering fixed bonus increases social welfare when the labor supply is thick but providing contingent bonus (i.e. paying only when drivers single-home) reduces the social welfare due to the result of platforms engaging in a price war on the supply side. \citet{lee_winner_2022} study whether there is a tipping over (i.e. devolving to monopoly) behavior in a ridesharing economy which has a presence of a negative same-side network effect. Using a differential game framework, they find that two-sided networks have a tendency to become a monopoly with the larger platform winning at the end, unless there is a product-differentiation in which case, the smaller platform can manage to survive.

\citet{zhang_two-sided_2022} study the effect of different wage schemes (fixed rate, dynamic rate, fixed wage) on the overall welfare of the economy and find that suitability of wage schemes depend upon whether the supply market is more competitive than the demand market. \citet{lian_labor_2022}, by using a queuing model, show that under competition, incumbent firms may stop contributing towards maintaining the labor pool causing the demand market to collapse. They also suggest that the incumbent platforms have incentive for tacit collusion to prevent entrants in the market under such conditions. \citet{hu_precommitments_2023} study the market behavior under different modes of precommitments of the market parameters (wage, price, and commission). They suggest that precommitting towards the market which is more competitive has worse outcomes for the platforms than not committing at all. Finally, \citet{wu_two-sided_2020} model platform competition in different modes according to whether the workers-customers decide simultaneously or sequentially. They find that in the sequential subgame mode, the duopoly tips over to a monopoly.

Most of the works outlined above model the interaction between customers and drivers as a simultaneous decision subgame i.e. in Nash Equilibrium, except for \citet{wu_two-sided_2020, bryan_theory_2019, chen_bonus_2022}. Like these works, we model the subgame between drivers and passengers using a sequential game where drivers move first. Our work is perhaps most similar to \citet{wu_two-sided_2020} and like them, we also obtain similar results concerning the nature of the equilibrium under a sequential driver-passenger subgame. In fact, \citet{bryan_theory_2019, zhang_inter-platform_2021, bai_can_2022, chen_bonus_2022, lee_winner_2022, zhang_two-sided_2022, lian_labor_2022} also, in one way or another, predict that under duopoly competition with multi-homing behaviors in a homogeneous market without any service-differentiation, the platform economy either devolves into a monopoly or competition drives entire market to one where no platform can earn any profit. While we re-derive similar results, we take it one step further by discussing conditions under which both platforms can co-exist in a duopoly equilibrium with multi-homing even when there is no service-differentiation. Our analyses points us towards a tacit collusion between two competing platforms. And while works such as \citet{lian_labor_2022} do hint of a possibility of collusion between ridesharing platforms, we explicitly discuss the types and nature of such collusions and obtain different collusion equilibriums.

\subsection{Collusion in Platform Economy}
Our work also ties to the literature surrounding collusion in platform economy. Here as well, we point readers to related surveys \citep{jullien_economics_2021, saattvic2018antitrust} instead of an in-depth review of our own and merely focus on the gap between existing literature and rideshare economy via some representative works. \citet{ruhmer_platform_2010} is one of the first works to investigate collusions in a two-sided markets and they argue that due to the cross-network externalities present in the network, two-sided collusion are typically harder to sustain in two-sided networks than in traditional settings. Similarly, \citet{evans_antitrust_nodate} argue that competition and antitrust policies for multi-sided platforms must be fundamentally different than traditional single-sided ones as the latter does not take into account the interdependent demands. \citet{lefouili_collusion_2020} expands upon \citet{ruhmer_platform_2010} by focusing on collusion in case of product-differentiation. They analyze cases of imperfect collusion and when one side of the market is allowed to multi-home. \citet{peitz_collusion_2022} also conduct a similar study with non-differentiated platforms and show different equilibriums for different discount factors in a repeated game setting. All of these studies only take into account cross externalities between two networks but do not consider same-side externalities like congestions and wait-costs, which ridesharing subscribes to. Similarly, these works model the interaction between the demand and supply side as a simultaneous game, which, as we justify later, may not be ideal in cases of rideshares.

\subsection{Mathematical Program Networks and Multilevel Programming}
Our model is based on mathematical program networks, a recent paradigm for unifying the pre-existing modelling approaches for interdependent mathematical programs, introduced by \citet{laine2024mathematical}. MPNs extend multilevel programming and provide a singular framework for specifying problems such as Nash Equilibrium problems, bi-, tri-, and multi-level programming problems, and Equilibrium Programs with Equilibrium Constraints (EPECs). Since MPNs are recent, there is no explicit literature surrounding it. However, its predecessor, i.e. multilevel programming has a rich background of its own which we now discuss. 

\citet{Cassidy1971EfficientDO} first used multilevel programming to model the flow of funds between the federal, state, and local levels. The field of multilevel programming then grew over the years with most of the early works being focused on linear trilevel problems \citep{Bard1984AnIO, ue1986hybrid, anandalingam1988mathematical, Benson1989OnTS}. Currently, it has been used to solve problems like supply chain modeling \citep{Han2017TrilevelDF, Fard2018ATL, Fard2018HybridOT}, cyber-defense \citep{Tian2019MultilevelPC}, energy scheduling \citep{luo2020energy}, and optimal control \citep{laine2023computation}. A survey of the works in multilevel programming till 2016 is from \citet{Lu2016MultilevelDA}, whereas a more recent literature review is from \citet{koirala2023monte}.

\section{Mathematical Program Networks (MPNs)}
\label{sec:MPNs}
As previously discussed, MPNs unify much of the existing modeling approaches for interdependent mathematical programs like Nash Equilibrium problems, multilevel programming problems, EPECs etc. A great advantage of MPN is that it not just allows for modeling these problems, but also presents a single algorithm to find their equilibrium points. We introduce MPNs and formally define their structure in this section. Much of the definition has directly been taken from \cite{laine2024mathematical}, the original work introducing MPNs, which we refer to our readers for additional details.

For some integer $k$, let $[k]$ define the index set $\{1, 2, \hdots, k\}$. Let $x_j$ denote the jth element of the vector $\textbf{x} \in \real^n$. For some index set $J \subset [n]$, let $[x_j]_{j \in J}$ denote the vector comprised of the elements of $\textbf{x}$ selected by the indices in $J$. 
\begin{definition}[Mathematical Program (MP)]
    A MP is a tuple $(f, C, J)$ that represents a mathematical programming problem of optimizing the vector of decision variables $\textbf{x} \in \real^n$ with respect to the cost function $f: \real^n \to \real$ over the feasible set $C \subset \real^n$. $J \subset [n]$ is a decision index set such that the sub-vector $[x_j]_{j\in J}$ is endogenous (or private) to the MP and all other variables are exogenous. As such, the MP is also said to be parameterized by the variables $[x_j]_{j \not\in J}$.
\end{definition}

\begin{definition}[MPN]
    A MPN is defined by the tuple $(\{MP^i\}_{i\in{[N]}}, E)$ and represents a directed graph comprised of $N$ MP nodes. MPN forms a network defined by the set of directed edges $E \subset [N] \times [N]$ and jointly optimizes a vector of decision variables $\textbf{x} \in \real^n$. The edge $(i,j)\in E$ iff node $j$ is a child of node $i$.
\end{definition}

\newcommand{\di}{ {\textbf{x}^D}^i }
\newcommand{\dis}{ {{(\textbf{x}^*)}^D}^i }
\newcommand{\dmi}{ {(\textbf{x}^*)^D}^{-i} }

We define the vector $\di$ to be the concatenation of all endogenous decision variables of the the MP $i$ and its descendants. Similarly, the vector ${\textbf{x}^D}^{-i}$ are all other exogenous decision variables which parameterize this MP. The solution graph of a MP node $i$ can then be defined as:

\begin{equation}
S^i := \left\{
\begin{aligned}
    \textbf{x}^* \in \real^n : \dis \in &\arg\min_{\di} &&f^i\left( \di, \dmi \right) \\
    &~~~~s.t. && \left( \di, \dmi \right) \in C^i \\
    &         && \left( \di, \dmi \right) \in S^j, (i, j) \in E
\end{aligned}
\right\}
\end{equation}

Where the argmin is taken in a local sense such that $S^i$ defines the set of all $\textbf{x}$ for which $\di$ are local minimizers to the optimization problem parameterized by $\dmi$. The endogenous decision variables of child nodes are assumed by the parent node(s), with the requirement that the parents are constrained by the solution graphs of the children.

\begin{definition} [Equilibrium]
    \label{def:equilibrium}
    A point $\textbf{x}^* \in \real^n$ is an Equilibrium of a MPN iff $\textbf{x}^*$ is an element of each node's solution graph. Specifically, $\textbf{x}^*$ is an equilibrium iff $\textbf{x}^* \in S^*$, where
    $$S^* := \bigcap_{i \in [N]} S^i$$
    The solution graph of any mathematical program is a subset of its feasible set. Therefore, $S^i \subset S^j \forall (i,j) \in E$.
\end{definition}

\subsection{Relationship between MPN and Multilevel programming}
A MPN extends multilevel programming approach. To be more specific, in multilevel programming, each level defines the order in which the decision is made, whereas in MPN, this definition is subsumed by the edges of the network. If $(i,j)\in E$ is an edge in the MPN, then the node $i$ decides first, inducing a rational reaction in the node $j$, which decides after. If node $i$ and $j$ are not connected i.e. $j$ is not reachable by following the directed edges from $i$ and vice versa, then we say that these nodes decide simultaneously.

\begin{figure}[h]
    \centering
    
    \begin{subfigure}[b]{0.45\textwidth}
        \centering
        \begin{tikzpicture}
            \node[draw, circle] (A) at (0,0) {A};
        
            \node[draw, circle] (B) at (3,0) {B};
        
        \end{tikzpicture}
        \caption{A and B decide simultaneously i.e. Nash Equilibrium.}
    \end{subfigure}
    \hfill
    \begin{subfigure}[b]{0.45\textwidth}
        \centering
        \begin{tikzpicture}
            \node[draw, circle] (A) at (0,0) {A};
        
            \node[draw, circle] (B) at (3,0) {B};
        
            \draw[->] (A) -- (B);
        \end{tikzpicture}
        \caption{A decides first followed by B i.e. Stackelberg Equilibrium.}
    \end{subfigure}
    
    \begin{subfigure}[b]{0.45\textwidth}
        \centering
        \begin{tikzpicture}
            \node[draw, circle] (A) at (0,0) {A};
        
            \node[draw, circle] (B) at (-1,-1) {B};
            \node[draw, circle] (C) at (1, -1) {C};
            \draw[->] (A) -- (B);
            \draw[->] (A) -- (C);
        \end{tikzpicture}
        \caption{A decides first followed by B and C i.e. mathematical program with equilibrium constraint (MPEC).}
    \end{subfigure}
    \hfill
    \begin{subfigure}[b]{0.45\textwidth}
        \centering
        \begin{tikzpicture}
            \node[draw, circle] (A) at (-1,0) {A};
        
            \node[draw, circle] (B) at (-1,-1) {B};
            \node[draw, circle] (C) at (1,0) {C};
            \node[draw, circle] (D) at (1,-1) {D};
            \draw[->] (A) -- (B);
            \draw[->] (C) -- (D);
        \end{tikzpicture}
        \caption{A and C decide simultaneously, followed by B and D. i.e. equilibrium program with equilibrium constraint (EPEC).}
    \end{subfigure}
    
    \caption{Representation of various multilevel programming approaches using MPN. }
\end{figure}
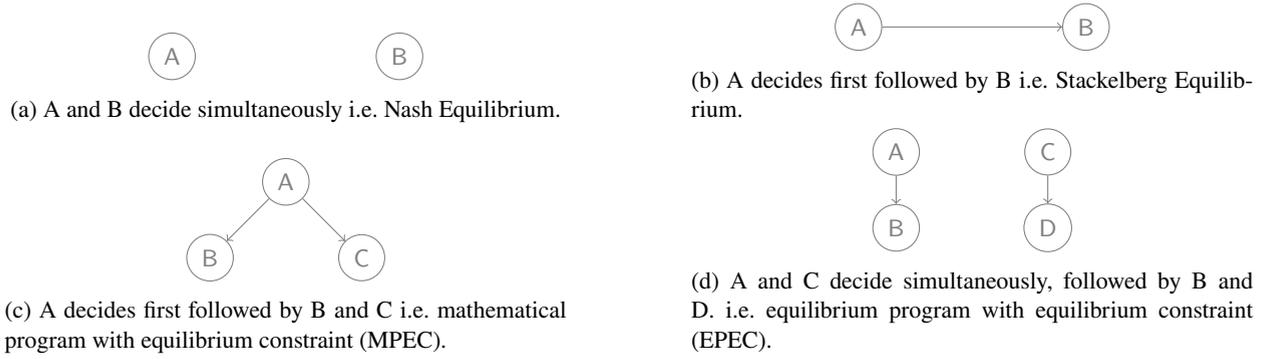

\section{Model}
\label{sec:model}

We consider a ridesharing duopoly with two platforms $U$ and $L$ and an outside option $P$. The platforms endogenously determine the wages and the price that they want to pay and exact from the drivers and the passengers respectively per mile, following which the drivers and the passengers engage in a sequential subgame to determine their allocation to the platform. Similar to \citet{wu_two-sided_2020, bryan_theory_2019, chen_bonus_2022}, we have chosen to model the interaction between drivers and riders as a sequential game where the drivers move first by deciding which platform to drive for followed by the passengers deciding on which platform to use given the price and the delay involved. We believe that this is a more accurate model of the ridesharing economy with both drivers and passengers multi-homing via an app in comparison to one where they decide on an allocation simultaneously. Since a rider can simultaneously open both platforms and decide who to ride with by considering the price and estimated wait-times (which is dependent upon the drivers around her location), it's generally up to the driver, who assuming has both apps installed as well, to decide which platform to choose for transporting the rider. This assumption is further bolstered by the fact that many real-world platforms are now providing upfront earning estimates for the drivers. This indicates that the drivers have an advantage on choosing a platform when both drivers and passengers are multi-homing and we base our model on this very justification. The overall sequence of our game can be summarized as:

\begin{enumerate}
    \item Two competing firms decide prices and wages for passengers and drivers respectively.
    \item Drivers schedule themselves amongst the platforms.
    \item Riders schedule themselves amongst the platform or any external choice like public transit.
\end{enumerate}

We use the following terminologies throughout our analysis.
\begin{align*}
U  &: \text{Platform U}.\\
L  &: \text{Platform L}\\
P  &: \text{Public Transit}\\
D  \in \real^+ &: \text{Average distance of a trip.}\\
\lambda \in \real^+ &: \text{Wait cost multiplier}\\
r^{u/l/p} \in \real^+ &: \text{Rate charged per mile to the passengers by }U/L/P\\
c^{u/l/p} \in \real^+ &: \text{Commission paid per mile to drivers by } U/L/P\\
p^{u/l/p} \in [0,1] &: \text{Proportions of passengers using } U/L/P\\
a^{u/l} \in [0,1] &: \text{Availability of drivers on } U/L\\
A \in [0, 1] &: \text{Total availability of the drivers across } U, L\\
a^{p} \in [0,1] &: \text{Availability of public transit}\\
g \in \real^+ &: \text{Gas cost per mile}\\
\end{align*}

\subsection{Passenger's Problem}
The passengers' problem is to collectively decide on an allocation that will minimize their total cost. For example, the cost for taking platform $U$ is the sum total of the price charged by $U$ and the waiting cost associated with $U$. For a trip of distance $D$, the total price paid is just the rate multiplied by distance i.e. $r^u \times D$, whereas, the waiting cost is given by the ratio of passengers taking $U$ to available drivers for $U$ scaled by the wait time multiplier $\lambda$. Therefore, the cost for taking $U$, say $\theta^U$, is given by:
$$\theta^U (p^u) := r^u D + \lambda \frac{p^u}{a^u}$$

And similar for $\theta^L, \theta^P$. Evidently, the cost of taking any platform is parameterized by the proportions of passengers opting for the platform and not by the rate charged or drivers available. This is because these variables are exogenously realized prior to the passenger's decision. The total cost for the passengers then becomes $p^u \theta^U(p^u) + p^l \theta^L (p^l) + p^p \theta^P(p^p)$ where $p^u, p^l, p^p$ are non-negative and sum to unity.

\subsection{Driver's Problem}
A driver driving for platform $U$ earns a commission of $c^u$ per mile but also has to spend gas money $g$ per mile. Considering that an average trip length is $D$, a driver for platform $U$ would then drive per-trip, an average length of $D \times \frac{p^u}{a^u}$. The total profit per trip for $U$ is then:

$$\gamma^U(a^u, p^u) := (c^u-g) D \times \frac{p^u}{a^u}$$

And similar for $\gamma^L(a^l, p^l)$. The drivers' problem is to maximize the total profit $a^u \gamma^U(a^u, p^u) + a^l \gamma^L(a^l, p^l)$ where $a^u, a^l$ are non-negative. We consider $p^{u/l/p}$ endogenous to this problem because these are induced rationally in the passenger level according to any decision made on the availability / allocation of drivers. Additionally, we impose a constraint on the allocation such that $a^u + a^l \le p^u + p^l$ and label it the \textit{matching constraint}. This establishes that total drivers in the system cannot exceed the total passengers. So if there are no passengers to be transported, there will be no driver participation either. This allows for the possibility of partial-covering of the demand-side (with the external option being $0$ for the drivers) and makes our model more expressive.

\subsection{Platform's Problem}
The platform problem is to simply maximize per trip revenue given by $p^{u/l} (r^{u/l} - c^{u/l}) D$ where the rates and commissions are non-negative.

\subsection{MPN formulation}
\label{subsec:MPN}
We are now ready to define our entire problem as a MPN. Prior to that, we make a few simplifying assumptions, without the loss of any generality. Specifically, we assume that average per-trip cost $D=1$, and that public transit has a high availability $a^p=1$. Then, our MPN consists of following nodes:

\begin{align}
    \label{mpn:u}
    U := \left\{ 
        \begin{aligned}
            \max_{r^u, c^u} ~~& p^u (r^u - c^u) \\
            & s.t. ~~ r^u, c^u \ge 0
        \end{aligned}
    \right\}
\end{align}

\begin{align}
    \label{mpn:l}
    L := \left\{ 
        \begin{aligned}
            \max_{r^l, c^l} ~~&p^l (r^l - c^l) \\
            & s.t. ~~r^l, c^l \ge 0~~~
        \end{aligned}
    \right\}
\end{align}

\begin{align}
    \label{mpn:d}
    D := \left\{ 
        \begin{aligned}
            \max_{a^u, a^l} ~~ &p^u (c^u - g) + p^l (c^l - g) \\
            & s.t. ~~0 \le a^u, a^l \le 1~~~ \\
            & ~~~~~~~a^u + a^l \le p^u + p^l
        \end{aligned}
    \right\}
\end{align}

\begin{align}
    \label{mpn:p}
    P := \left\{ 
        \begin{aligned}
            \min_{p^u, p^l, p^p} ~~ &p^u\left(r^u + \lambda \frac{p^u}{a^u}\right) + p^l\left(r^l + \lambda \frac{p^l}{a^l}\right) + p^p\left(r^p + \lambda p^p\right)\\
            &~~s.t.~~~~ 0 \le p^u, p^l, p^p \le 1\\
            &~~~~~~~~~~ p^u + p^l + p^p = 1
        \end{aligned}
    \right\}
\end{align}

And the edges are defined such that the network in figure \ref{fig:MPN} is obtained.
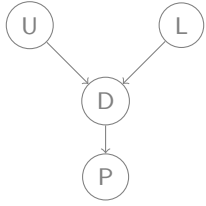
\begin{figure}[h]
    \centering
    \begin{tikzpicture}
        \node[draw, circle] (U) at (-1,1) {U};
    
        \node[draw, circle] (L) at (1,1) {L};
        \node[draw, circle] (D) at (0,0) {D};
        \node[draw, circle] (P) at (0,-1) {P};
        \draw[->] (U) -- (D);
        \draw[->] (L) -- (D);
        \draw[->] (D) -- (P);
    \end{tikzpicture}
    \caption{MPN for the ridesharing duopoly problem. U and L decide simultaneously, followed by D followed by P.}
    \label{fig:MPN}
\end{figure}

The problem models a sequential game consisting of three steps. At first step, the platforms choose their rates and commissions $r^{u/l}, c^{u/l}$ simultaneously. Their choice of these parameters induces, at second step, the drivers to choose their availability (aka schedule themselves) across the two platforms. After both platforms and the drivers commit their decisions, the passengers, at the end, decide how to allocate themselves between the two platforms and the public transit. We assume that each of these actors are rational and their choice of action at each step maximizes (or minimizes) their utility (or cost) provided that each subsequent players also choose their actions in a rational manner.

\subsection{Externalities}
\subsubsection{Cross-Side Externalities}
Even though our model is parsimonious, it allows for different interactions across the two networks. For example, as the allocation of the drivers for a particular platform  (say $U$) increases, the waiting cost for passengers in the corresponding platform ($\lambda \frac{p^u}{a^u}$) decreases and the passengers are incentivized to use the platform more often. Similarly, although it's not clear from the objective function of the drivers itself, as more passengers want to use a platform, more drivers are incentivized to drive for it. For example, in a scenario where no passenger wants to use platform $L$ (perhaps due to very high per-mile cost) but wants to use platform $U$, drivers are incentivized to drive for $U$ to earn associated profit. If they don't drive for $U$ (and drive for $L$ instead), then, passenger wait time skyrockets and they end up choosing the transit option leaving drivers without any profit at all.

\subsubsection{Same-Side Externalities}
For the passengers, the wait cost increases as the proportion of passengers using the platform increases, so the model readily shows negative same-side effect on the passenger side. For the drivers, this effect manifests differently. In a partially-covered market i.e. $a^u + a^l < p^u + p^l$, as the number of drivers for any platform increases, the collective driver profit may increase (because drivers now serve the passengers who would've taken transit otherwise). But when the supply market is fully covered or overcrowded i.e. $a^u + a^l \ge p^u + p^l$, the \textit{matching constraint} forces a \textit{decrease} in the number of drivers. 

\subsection{Single-Homing vs Multi-Homing Interpretation}
Our model has valid interpretations under both single-homing  and multi-homing regimes. Under the SH interpretation, it is a platform allocation problem where drivers and passengers are allocating themselves between multiple platforms exclusively. But under the MH interpretation, it is a \textit{preference} allocation problem where the drivers and passengers use all platforms with differing preferences. Concretely, $p^u = p^l = p^p = \frac{1}{3}$ can either be interpreted as each third of the passengers using platforms $U/L/P$ exclusively or all passengers using the respective platforms with equal preference (i.e. indifferent choice).

\subsection{Optimism}
\label{subse:optimisim}
In our model, the drivers are optimistic players. Crucially, this implies that under conditions when participating in the market is the same as not-participating, then the drivers would always participate. This is a reasonable assumption reflective of the real-world because a) the drivers are incentivized not only because of the wages but also due to tips provided by the passengers (which is a significant aspect of the ridesharing economy, albeit not modeled explicitly in this work), b) the sense of involvement when the alternative is to not engage in any employment opportunity at all may be more desirable, and c) drivers may overestimate net profits from their revenue in the face of rising automobile maintenance costs.

\section{Equilibrium Analysis}
\label{sec:equilibrium}
We begin by defining a solution to the MPN defined in subsection \ref{subsec:MPN}. 
\begin{definition}[Passenger's Rational Choice]
    If $f^P(p^u, p^l, p^p)$ is the parameterized objective function of the passengers defined in (\ref{mpn:p}) and $\mathcal{C}^P$ is the feasible set defined by (\ref{mpn:p}), then $S^P := ({p^u}^*, {p^l}^*, {p^p}^*)$ is the passenger's rational choice for the given exogenous variables ($r^u, c^u, r^l, c^l, a^u, a^l$) iff $\forall (p^u, p^l, p^p) \in \mathcal{C}^P, f^P({p^u}^*, {p^l}^*, {p^p}^*) \le f^P(p^u, p^l, p^p).$
\end{definition}

\begin{definition}[Driver's Rational Choice]
    If $f^D(a^u, a^l, p^u, p^l, p^p)$ is the parameterized objective function of the drivers defined in (\ref{mpn:d}) and $\mathcal{C}^D$ is the feasible set defined by (\ref{mpn:d}), then $S^D(r^u, c^u, r^l, c^l) := ({a^u}^*, {a^l}^*, {p^u}^*, {p^l}^*, {p^p}^*)$ is the driver's rational choice for the given exogenous variables $(r^u, c^u, r^l, c^l)$ iff $\forall (a^u, a^l, p^u, p^l, p^p) \in \mathcal{C}^D \cap S^P,~  f^D({a^u}^*, {a^l}^*, {p^u}^*, {p^l}^*, {p^p}^*)\ge$ $f^D(a^u, a^l, p^u, p^l, p^p)$
\end{definition}

\begin{definition}[Solution]
    With $\mathcal{C}^U, \mathcal{C}^L$ as the feasible set for the problem defined in (\ref{mpn:u}, \ref{mpn:l}), we define $S := ({r^u}^*, {c^u}^*, {r^l}^*, {c^l}^*, {a^u}^*, {a^l}^*, {p^u}^*, {p^l}^*, {p^p}^*)$ to be the solution of the entire MPN defined in subsection (\ref{subsec:MPN}) iff $({r^u}^*, {c^u}^*,{r^l}^*, {c^l}^*) \in \mathcal{C}^U \cap \mathcal{C}^L$ is the Nash-Equilibrium of the maximization objectives defined in (\ref{mpn:u}, \ref{mpn:l}) and $({a^u}^*, {a^l}^*, {p^u}^*, {p^l}^*, {p^p}^*) \in S^D$. In other words, a solution is the Equilibrium (definition \ref{def:equilibrium}) of the MPN defined in subsection \ref{subsec:MPN}.
\end{definition}

Similar to \citep{wu_two-sided_2020}'s approach, we also define Non-Trivial and Trivial Equilibrium solutions as follows:

\begin{definition}[Non-Trivial and Trivial Equilibrium Solutions]
    A solution $S := ({r^u}^*, {c^u}^*, {r^l}^*, {c^l}^*, 
    {a^u}^*, {a^l}^*, {p^u}^*, {p^l}^*, {p^p}^*)$ is a Non-Trivial Equilibrium solution iff both platforms earn non-zero profits i.e. ${p^u}^*({r^u}^*-{c^u}^*) > 0$ and ${p^l}^*({r^l}^*-{c^l}^*) > 0$. If any one or both of the platforms earn a zero profit, we say $S$ to be a Trivial Equilibrium solution.
\end{definition}

In addition to these, we also define a Strong Equilibrium solution as follows:

\begin{definition}[Strong Non-Trivial Equilibrium Solutions]
    \label{def:sntes}
    A Non-Trivial Equilibrium Solution $S$ is a Strong Non-Trivial Equilibrium solution if there is some participation from passengers in public transit as well, i.e. ${p^p}^* \not < 0$. Alternatively, we define it as an Equilibrium condition where the constraint $p^{p^*} \ge 0$ is inactive.
\end{definition}

We are particularly interested in finding Non-Trivial Equilibrium solutions for our problem, therefore, we begin first by assuming that there indeed exists a Non-Trivial Equilibrium. In fact, we take this assumption a step further by saying that there actually exists a Strong Non-Trivial Equilibrium Solution and guide our analysis towards finding it. As we show later in this section, this assumption allows us to derive when a Strong Non-Trivial Solution degenerates to either a Non-Trivial (see lemma \ref{lemma:equiv}) or a Trivial Equilibrium Solution and thus, does not result in a loss of generality.

\begin{assumption}
    \label{ref:non-trivial}
    There exists a Strong Non-Trivial solution $S^*$ 
\end{assumption}

The direct consequence of Assumption \ref{ref:non-trivial} is that for $S^*$, we must have non-zero allocation of passengers for all platforms including the public transit, i.e. $p^u, p^l, p^p > 0$. Since $p^u+p^p+p^l = 1$, this also implies $p^u, p^l, p^p < 1$. Similarly, we must also have non-zero allocation of drivers for both platforms, otherwise for $\lambda > 0$, there would be infinite wait cost and the passengers would choose alternative platform instead. Therefore, $0 < a^u, a^l < 1$, which also further implies that driving does not incur losses so $c^{u/l} \ge g$. Finally, it must also be the case that the platforms do not incur losses per mile i.e. $r^{u/l} - c^{u/l} \ge 0$. These can be summarized as:
\begin{subequations}
    \label{eq:assumptions}
    \begin{align}
        \label{assm:nontrivpassengers}
        0 < p^{u/l/p} &< 1 \\
        \label{assm:nontrivdrivers}
        0 < a^{u/l} &< 1 \\
        \label{assm:profitdrivers}
        c^{u/l} &\ge g \\
        \label{assm:profitplatform}
        r^{u/l} &\ge c^{u/l}
    \end{align}
\end{subequations}

\subsection{Passenger's Rational Choice}
For the Strong Non-Trivial Solution $S^*$ by (\ref{assm:nontrivpassengers}), there is no action of the constraint $0 \le p^{u/l/p} \le 1$ on the passenger's choices. Therefore, their optimization problem simply degenerates to:

\begin{subequations}
\begin{alignat}{2}
    \label{passenger:objective2}
    &~~~~~~~~~~~~~~~~~~~~\min_{p^u, p^l, p^p} p^u\left(r^u + \lambda \frac{p^u}{a^u}\right) + p^l\left(r^l + \lambda \frac{p^l}{a^l}\right) + p^p\left(r^p + \lambda p^p\right)\\    &~~~~~~~~~~~~~~~~~~~~~~~~~~~~~~~~~~~~~~~~~~~~~~p^u + p^l + p^p = 1
\end{alignat}
\end{subequations}
\begin{remark}
    The optimization problem of the passengers is strictly convex.
\end{remark}
\begin{proof}
    The Hessian of the objective function with respect to the variables is given by: 
    $$\nabla^2 f^P = \begin{bmatrix}
    \frac{2\lambda}{a^u} & 0 & 0 \\
    0 & \frac{2\lambda}{a^l} & 0 \\
    0 & 0 & 2\lambda
\end{bmatrix}$$
    Which is a diagonal matrix. Thus, for $a^u, a^l, \lambda >0$ all eigenvalues of the Hessian are strictly positive and the problem is strictly convex.
\end{proof}

This implies that a minimizer exists and can be found by using the First Order Necessary Conditions (FONCs). By substituting $p^p = 1-p^u-p^l$ in the objective (\ref{passenger:objective2}), we first obtain the following FONCs:

\begin{subequations}
  \begin{align}
    \label{eq:fonc1}
    (r^l -r^u) &= 2\lambda \left( \frac{p^l}{a^l} - \frac{p^u}{a^u}\right) \\
    \label{eq:fonc2}
    (r^l -r^p) &= 2\lambda \left( 1-p^u-p^l - \frac{p^u}{a^u}\right)
\end{align}  
\end{subequations}

Solving equations (\ref{eq:fonc1}), (\ref{eq:fonc2}), we obtain that for $S^*$; $p^u, p^l, p^p ~(= 1-p^u-p^l)$ are given as:
\begin{subequations}
\begin{align}
    \label{eq:pu}
    p^u &= \frac{2\lambda a^u + a^l a^u (r^l - r^u) + a^u (r^p-r^u)}{2\lambda (a^l + a^u + 1)} \\ 
    \label{eq:pl}
    p^l &= \frac{2\lambda a^l + a^l a^u (r^u - r^l) + a^l (r^p - r^l)}{2\lambda (a^l + a^u + 1)} \\
    \label{eq:pp}
    p^p &= \frac{2\lambda + a^l (r^l - r^p) + a^u (r^u -r^p)}{4\lambda}
\end{align}
\end{subequations}

\subsubsection{Pricing conditions for $S^*$ to degenerate into a Trivial Equilibrium Solution}
We start by considering platform $U$ and results for platform $L$ will follow by symmetry. When $p^u \le 0$, we have that no passengers take the platform $U$. This trivially happens when no drivers drive for platform $U$ i.e. from equation (\ref{eq:pu}), when $a^u=0, p^u=0$. But even when we consider that platform $U$ has some supply presence $a^u > 0$ in the most favorable scenario i.e. there is no competition in the supply side ($a^l = 0$), we still get that $p^u \le 0$ when:

\begin{align*}
    p^u  = \frac{2\lambda a^u + a^u (r^p-r^u)}{2\lambda ( a^u + 1)} &\le 0 \\
    2\lambda + r^p -r^u &\le 0 \\
    2\lambda + r^p &\le r^u
\end{align*}

This shows that even in the most favorable supply market, if rate $r^u$ exceeds $r^p + 2\lambda$, platform $U$ does not attract any demand and, thus the equilibrium degenerates into a Trivial Equilibrium. Therefore, we can infer that in a Strong Non-Trivial Equilibrium, it must be the case for both platforms that their rates are upper bounded as:
\begin{align}
    r^{u/l} \le r^p + 2\lambda
\end{align}
From these equations, we also infer that if the rate of the public transit $r^p \le g - 2\lambda$, for demand to exist, $r^{u/l} \le g$. But charging less than gas cost and paying more than it (to attract drivers), does not yield a profit to the platforms so no Non-Trivial Equilibrium exists. In general, we assume that the rate of public transit (which also incorporates the inconvenience associated with it) is high enough that this is never the case.

\subsubsection{Pricing conditions for $S^*$ to degenerate into a Non-Trivial Equilibrium Solution}
This happens when no one takes public transit i.e. $p^p \le 0$:
\begin{align*}
    p^p = \frac{2\lambda + a^l (r^l - r^p) + a^u (r^u -r^p)}{4\lambda} &\le 0\\ 
    a^lr^l + a^ur^u -r^p(a^l+a^u) &\le -2\lambda\\
    a^lr^l + a^ur^u &\le r^p(a^l+a^u) - 2\lambda
\end{align*}
Assuming $r^u < r^l$, w.l.o.g., we get:
\begin{align*}
    (a^l + a^u) r^u \le r^p (a^l + a^u) - 2\lambda \\
    r^u \le r^p - \frac{2\lambda}{a^l + a^u}
\end{align*}
This indicates that for any platform that prefers to maximize its profit, there is simply no need to set their price below $r^p - 2\lambda / (a^l + a^u)$ even if they wanted to attract a higher number of passengers, i.e. the rate of the platform is lower bounded as:
\begin{align}
\label{ineq:rate_lower_bound}
    r^{u/l} \ge r^p - \frac{2\lambda}{a^u + a^l}
\end{align}

\begin{lemma}
\label{lemma:equiv}
    Any Non-Trivial Equilibrium Solution is also effectively a Strong Non-Trivial Equilibrium.
\end{lemma}
\begin{proof}
From inequality (\ref{ineq:rate_lower_bound}) it follows that platforms who seek to maximize their profits have no reason to set their rates beyond the threshold where the constraint $p^p \ge 0$ is engaged even if they aim for a full market capture. Therefore, by definition (\ref{def:sntes}), a Non-Trivial Equilibrium Solution is always a Strong Non-Trivial Equilibrium.
\end{proof}

\subsection{Driver's Rational Choice}
By substituting the rational choices of the passengers obtained in equations (\ref{eq:pu}, \ref{eq:pl}), we obtain the optimization problem for the drivers as:
\begin{subequations}
\label{eq:driver2}
    \begin{align}
    \label{eq:driver_opt2}
    \max_{a^u, a^l} &\frac{2\lambda a^u + a^l a^u (r^l - r^u) + a^u (r^p-r^u)}{2\lambda (a^l + a^u + 1)} (c^u - g) + \frac{2\lambda a^l + a^l a^u (r^u - r^l) + a^l (r^p - r^l)}{2\lambda (a^l + a^u + 1)} (c^l -g) \\
    &s.t.~~ 0 \le a^u + a^l \le p^u + p^l \\
    &~~~~~~ r^u, r^l \le r^p + 2\lambda
\end{align}
\end{subequations}

In Strong Non-Trivial Equilibrium Solution $S^*$, there exists an optimal allocation $\opt{a^u}, \opt{a^l} > 0$ by assumption. We let $A = \opt{a^u} + \opt{a^l} > 0$. For the purpose of following analysis, we do not focus on what an optimal $A$ is and instead, argue about what the optimal allocation of $a^u, a^l$ in $A$ is. By substituting $a^l = A-a^u$, in the optimization objective (\ref{eq:driver_opt2}), we obtain the drivers allocation function:

\begin{subequations}
    \label{eq:alloc_obj}
   \begin{align}
    f^D(a^u) = &\frac{2\lambda a^u + (A-a^u) a^u (r^l - r^u) + a^u (r^p-r^u)}{2\lambda (A + 1)} (c^u - g) \\
    &+ \frac{2\lambda (A-a^u)  + (A-a^u)  a^u (r^u - r^l) + (A-a^u) (r^p - r^l)}{2\lambda (A + 1)} (c^l -g)
\end{align} 
\end{subequations}

Which is quadratic in $a^u$. The optimization problem for drivers now change as:
\begin{subequations}
\label{eq:driver3}
    \begin{align}
    \label{eq:driver_opt3}
    \max_{a^u} &f^D(a^u) \\
    &s.t.~~ 0 \le a^u \le A \\
    &~~~~~~ r^u, r^l \le r^p + 2\lambda
\end{align}
\end{subequations}

We now prove the following result:

\begin{theorem}
    There exists no mixed strategy profile for drivers allocation strategy ${a^u}$ which is strictly better than a pure strategy. In other words, drivers prefer monopoly.
\end{theorem}

\begin{proof}
    The Hessian of the driver's allocation objective (\ref{eq:alloc_obj}) is given by:

    \begin{align}
        \nabla^2 f^D = \frac{(c^l-c^u)(r^l-r^u)}{\lambda (A+1)}
    \end{align}

    For any mixed strategy to be strictly better than a pure strategy for a quadratic objective, there must exist an $0 < {a^u}^* < A$ which is a maximizer for $f^D$. In other words, $f^D$ must be a concave function with its peak in the open interval $(0, A)$. For $f^D$ to be concave, $\nabla^2 f^D < 0$. Which is to say, it must be either of the following:
    \begin{enumerate}
        \item $c^u > c^l$ and $r^l > r^u$ or,
        \item $c^l > c^u$ and $r^l < r^u$
    \end{enumerate}

    We begin by assuming that such mixed strategy exists and show contradiction. Under such assumption, the maximizer is obtained when $\frac{\partial f^D}{\partial a^u} = 0$ which gives:
    \begin{align}
    \label{eq:mainineq}
        {a^u}^* = \frac{(c^l-c^u)(Ar^l-Ar^u+2\lambda + r^p) - c^l r^l+ c^u r^u +g (r^l - r^u)}{2(r^u-r^l)(c^u-c^l)}
    \end{align}

    \textbf{CASE I: } Assume that $c^u > c^l$, then $r^l > r^u$. Multipling both numerator and denominator in (\ref{eq:mainineq}) by $-1$,
    $${a^u}^* = \frac{(c^u-c^l)(Ar^l-Ar^u+2\lambda + r^p) + c^l r^l- c^u r^u - g (r^l - r^u)}{2(r^l-r^u)(c^u-c^l)}$$
    Since $c^l \ge g$ and $r^l-r^u > 0$, we can increase the $g$ in the numerator to $c^l$, thereby increasing the overall negative term and decreasing the RHS.
    
    $${a^u}^* > \frac{(c^u-c^l)(Ar^l-Ar^u+2\lambda + r^p) + c^l r^l- c^u r^u - c^l (r^l - r^u)}{2(r^l-r^u)(c^u-c^l)}$$
    
    $${a^u}^* > \frac{(c^u-c^l)(Ar^l-Ar^u+2\lambda + r^p) - (c^u-c^l)r^u}{2(r^l-r^u)(c^u-c^l)}$$
    
    $${a^u}^* > \frac{Ar^l-Ar^u+2\lambda + r^p - r^u}{2(r^l-r^u)}$$
    
    Since $r^l \le 2\lambda + r^p$, can decrease $2\lambda + r^p$ to $r^l$ to further decrease the RHS.
    
    $${a^u}^* > \frac{Ar^l-Ar^u+r^l - r^u}{2(r^l-r^u)}$$
    
    $${a^u}^* > \frac{A+1}{2}$$

    Since, $A < 1$, can decrease 1 to $A$ to finally obtain:
    $${a^u}^* > A$$
    Which is a contradiction since ${a^u}^* < A$.
    
    \textbf{CASE II: } Assume that $c^l > c^u$, then $r^u > r^l$. From equation (\ref{eq:mainineq}), multiplying both numerator and denominator by -1:
    
    $${a^u}^* = \frac{(c^u-c^l)(Ar^l-Ar^u+2\lambda + r^p) + c^l r^l- c^u r^u + g (r^u - r^l)}{2(r^u - r^l)(c^l-c^u)}$$
    Since $r^u - r^l > 0, c^u \ge g$ can increase $g$ to $c^u$ to increase the positive term further, thereby increasing the RHS:
    
    $${a^u}^* < \frac{(c^u-c^l)(Ar^l-Ar^u+2\lambda + r^p) + c^l r^l- c^u r^u + c^u (r^u - r^l)}{2(r^u - r^l)(c^l-c^u)}$$
    $${a^u}^* < \frac{(c^u-c^l)(Ar^l-Ar^u+2\lambda + r^p) + c^l r^l- c^u r^l}{2(r^u - r^l)(c^l-c^u)}$$
    $${a^u}^* < \frac{(c^l-c^u)(r^l - Ar^l + Ar^u - 2\lambda -r^p)}{2(r^u - r^l)(c^l-c^u)}$$
    $${a^u}^* < \frac{(r^l - Ar^l + Ar^u - 2\lambda -r^p)}{2(r^u - r^l)}$$
    
    Since $r^u \le 2\lambda + r^p$, we can increase the RHS by decreasing the negative $-2\lambda - r^p$ term to $-r^u$.
    
    $${a^u}^* < \frac{r^l - r^u + A(r^u - r^l)}{2(r^u - r^l)}$$
    $${a^u}^* < \frac{(r^u-r^l)(A-1)}{2(r^u - r^l)}$$
    $${a^u}^* < \frac{(A-1)}{2}$$
    $${a^u}^* < 0$$
    Which is a contradiction since ${a^u}^* > 0$.
\end{proof}

Hence, we have shown that no mixed strategy exists for the driver's allocation that is strictly better than any of the pure strategies. This corresponds to the \textit{tipping-over} behavior that is generally observed in a two-sided market with positive cross-side externalities. Under such tipping-over, the solution degenerates into a Trivial Equilibrium and drivers' participation can be calculated as follows:

\begin{lemma}
    Under a monopoly, drivers participate to their full extent in the market.
\end{lemma}

\begin{proof}
    Assume $a^u = A$ and $a^l = 0$. Then, the drivers' optimization objective in terms of market presence becomes:
    $$\tilde f ^D (A) = \frac{2\lambda A + A (r^p - r^u)}{2\lambda (A+1)} (c^u-g)$$
    $$= \left[ \frac{2\lambda + r^p - r^u}{2\lambda}(c^u-g)\right] \frac{A}{A+1}$$
    The derivative of this function is:
    $$\tilde f'^D = \left[ \frac{2\lambda + r^p - r^u}{2\lambda}(c^u-g)\right] \frac{1}{(A+1)^2}$$
    Which is always positive for $r^u < 2\lambda + r^p$ and $c^u >g$. Therefore, drivers participate with maximum $A$ only bounded by the presence of passengers, i.e. $A = p^u + p^l$. Similar results follow when $a^l = A, a^u = 0$.
\end{proof}

In order for a Non-Trivial Equilibrium solution to exist, as no strictly dominant mixed strategy exists, there must now be a weakly dominant mixed strategy ${a^u}^*$ which is at least as good as either of the pure strategies i.e. $f^D({a^u}^*) = f^D(0)$ or $f^D({a^u}^*) = f^D(A)$. Furthermore, the derivative $\frac{\partial f^D}{\partial a^u}$ at ${a^u}^*$ must be exactly 0, otherwise the drivers could just deviate in the direction of increasing derivative to increase their payoff thereby establishing ${a^u}^*$ to be non-optimal. We show that this can only happen if the payoff function is a constant response function.

\begin{lemma}
    A quadratic function $f(x) = cx^2 + dx + e$ with two points $a \ne b$ and $f(a)=f(b)$ and $\frac{\partial f}{ \partial x} = 0$ at either $x=a$ or $x=b$ is a constant response i.e. $f(x) = e$.
\end{lemma}
\begin{proof}
    Assume $c \ne 0$. When $f(a)=f(b)$, we have that $ca^2 + da = cb^2 + db$. Assume $f'(a) = 0$ which gives $a = -\frac{d}{2c}$. Substituting, we get:
    \begin{align*}
        c \frac{d^2}{4c^2} - \frac{d^2}{2c} &= cb^2 + db \\
        d^2 + 2d cb + 2c^2 b^2 &= 0
    \end{align*}
    Which yields that:
    $$d = \frac{-2cb \pm \sqrt{4c^2b^2-8c^2b^2}}{2} = \frac{-2cb\pm\sqrt{-4c^2b^2}}{2}$$
    Which is a contradiction. Similar results follow when $f'(b)=0$. This shows that $c=0$. But when $c=0$, either  $f'(a)=0$ or $f'(b)=0$ yields $d=0$.
\end{proof}

Thus, in a Strong Non-Trivial Equilibrium solution, the response of the allocation function for drivers is a constant response, i.e. $f^D(a^u) = C$ for some constant $C$. This allows us to establish two results.
\begin{enumerate}
    \item The Hessian of the allocation function is 0, i.e.
    \begin{align}
        \nabla^2 f^D(a^u) = \frac{(c^l-c^u)(r^l-r^u)}{\lambda (A+1)} &= 0 \\
        (c^l-c^u)(r^l-r^u) &= 0
    \end{align}
    \item The pure strategies $\opt{a^u} = 0$ and $\opt{a^u} = A$ have the same payout, i.e. $f^D(0) = f^D(A)$, which yields:    
    \begin{align}
    \label{eq:pure_strategy_balance}
        (2\lambda + r^p - r^l)(c^l - g) = (2\lambda + r^p - r^u)(c^u-g)
    \end{align}
\end{enumerate}

These, in turn, are satisfied iff:
\begin{subequations}
\label{eq:conditions}
\begin{align}
    \label{eq:clecu}
     c^l &= c^u \ne g, r^l = r^u \\
     \label{eq:rleru}
     r^l &= r^u \ne 2\lambda + r^p, c^l = c^u \\
     \label{eq:clecue}
     c^l &= c^u = g, \text{and~}r^l, r^u \text{~unconstrained} \\
     \label{eq:rlerue}
     r^l &= r^u = 2\lambda + r^p, \text{and~} c^l, c^u \text{~unconstrained}
\end{align}    
\end{subequations}

Out of these four conditions, (\ref{eq:rlerue}) pushes the market participation of the passengers to 0 and no-one opts to use any platforms i.e. the solution degenerates to a Trivial Equilibrium solution, and thus, is of no interest. At this point, we further make an assumption that when drivers have a constant payoff i.e. equations (\ref{eq:conditions}) are satisfied, the drivers participate equally in both platforms:

\begin{assumption}
    In a Non-Trivial Equilibrium Solution, $\opt{a^l} = \opt{a^u} = \frac{A}{2}$
\end{assumption}

\subsection{Platform's Rational Choice and the Strong Non-Trivial Equilibrium Solution}
\label{subsec:platform}
The conditions we outlined in equations (\ref{eq:conditions}) show us that a Non-Trivial Equilibrium Solution is only possible if both platforms engage in collusion of either two kinds:

\subsubsection{Double-Sided Collusion}
By equations (\ref{eq:clecu}) and (\ref{eq:rleru}), we see that if both platforms collude to make both their rates and their commission equal, they can induce indifference in drivers, $\opt{a^u} = \opt{a^l} = \frac{A}{2}$ and can both earn non-zero profits. Under this condition, the participation of the drivers can be calculated as follows:


\begin{lemma}
    Under a double-sided collusion, drivers participate to their full extent in the market.
\end{lemma}

\begin{proof}
    We have $\opt{a^u} = \frac{A}{2}$, $\opt{a^l} = \frac{A}{2}$, $\opt{r^l}=\opt{r^u}=r$, and $\opt{c^l}=\opt{c^u}=c$. Then, the driver's optimization objective in terms of market presence becomes:
    $$\bar f ^D (A) = \frac{2\lambda A + A (r^p - r)}{2\lambda (A+1)} (c-g)$$
    The derivative of this function with respect to $A$ is:
    $$\bar f'^D = \left[ \frac{2\lambda + r^p - r}{2\lambda}(c-g)\right] \frac{1}{(A+1)^2}$$
    Which is always positive for $r < 2\lambda + r^p$ and $c >g$. Therefore, drivers participate with maximum $A$ only bounded by the presence of passengers, i.e. $A = \opt{p^u} + \opt{p^l}$.
\end{proof}

\subsubsection{Single-Sided Collusion}
By equation (\ref{eq:clecue}), if the platforms collude to keep their commission the bare minimum, i.e. $\opt{c^u} = \opt{c^l} = g$, we see that they need not collude on the supply side and can set their rates freely. Under this collusion scheme, the payoff of the drivers is 0 but due to our assumption of optimism outlined in subsection (\ref{subse:optimisim}), they participate fully in the market.

\begin{lemma}
    \label{lemma:single-side}
    Under single-sided collusion, drivers participate to their full extent in the market.
\end{lemma}

\begin{proof}
    As above.
\end{proof}

\subsubsection{Perfect competition and grim trigger}
If the platforms are engaged in a perfect competition, a Non-Trivial Equilibrium is impossible. Assume a Non-Trivial Equilibrium condition where platforms are engaged in a single-sided collusion, i.e. $\opt{c^l} = \opt{c^u} = g$, the participation from passengers in this market is given by:
$$\opt{p^u} = \opt{p^l} = \frac{4\lambda A + A^2 (r^l - r^u) + 2A (r^p-r^u) }{8\lambda (A+1)}$$

According to condition (\ref{eq:pure_strategy_balance}), if a platform, say $U$, increases its commission by a small amount, ${c^u}' = \opt{c^u} + \epsilon$, then it's more beneficial for drivers to choose the pure strategy $\opt{a^u}=A$. When this happens, the new passenger participation for $U$ becomes:
\begin{align*}
    {p^u}' &= \frac{2\lambda A + A(r^p - r^u)}{2\lambda (A+1)} \\
\end{align*}
It can be shown that ${p^u}' > \opt{p^u}$ and since the platform's profit is given by $p^u(r^u - c^u)$, this is an optimal choice in a perfect competition. Similarly, for a double-sided collusion, from the same condition (\ref{eq:pure_strategy_balance}), either of the platforms can entice the entire supply side, i.e. the drivers by either decreasing the rates by a little or increasing the commission by little and obtain a greater profit. This, in turn, instigates the other platform to either provide better rates and payoffs for a monopoly profit or at least match the rates / commission for a non-zero profit. Therefore, in case a grim trigger is pressed, the equilibrium will degenerate to one where $\opt{r^u} = \opt{r^l} = \opt{c^u} = \opt{c^l}$ and neither platform will earn any profit.

\subsubsection{Optimal Collusion}
For a double-sided collusion with $\opt {r^u}= \opt{r^l}$ and $\opt{c^u}=\opt{c^l}$, the optimal equilibrium must have $\opt{c^u}=\opt{c^l}=g$ as the lowest commission yields the highest profits for both platforms. But, from single-sided collusion perspective, once the collusion becomes as such, there is no need for the platforms to collude on the demand side again. Therefore, we argue that single-sided collusion is the natural outcome of this game. In general, these sequence of events best describe our argument:
\begin{enumerate}
    \item Platform may begin by engaging on a perfect competition but both get a zero payoff due to the price war.
    \item Platforms realize that tacitly colluding on both supply and demand side maintains a non-zero payoff for both.
    \item Platforms decrease the commission to the bare minimum to maximize their double-sided collusion equilibrium.
    \item Once the commission is reduced to the bare minimum, platforms notice that engaging in a competition on the demand side, does not actually cause the equilibrium to collapse. This induces a stable single-sided collusion.
\end{enumerate}



\section{Conclusion}
\label{sec:conclusion}
In this work, we modeled a ridesharing game between duopoly platforms, drivers, and passengers using a mathematical program network approach. We were able to show that this model demonstrates the crucial externalities present in the actual ridesharing economy i.e. a negative same-side externality and a positive cross-side externality. We then analyzed the equilibrium by first assuming the existence of a Strong Non-Trivial Equilibrium solution (which we showed to be the same as a Non-Trivial Equilibrium solution) and then deriving the conditions of its existence. The major condition for the existence of a Non-Trivial Equilibrium was the participation decision of the drivers. We found that, in general, there is no strictly dominant mixed strategy for drivers in a Non-Trivial Equilibrium condition, i.e. drivers prefer monopoly. This is in line with the \textit{tipping over} behavior generally found in the markets with positive cross-externalities. We then sought to find a weakly dominant mixed strategy for the drivers which necessitated that the allocation function of the drivers be a constant response function. However, we argued that such indifferent constant response could only be obtained under certain conditions which naturally indicated a potential for tacit collusion between the two platforms. Lastly, we also showed that in case of a perfect competition, a Non-Trivial Equilibrium solution was not possible. This result validates the findings of existing literature.

Our main results follow from subsection \ref{subsec:platform}. In particular, we have shown that since the platforms are rational agents, and seek to maximize their profits, they can do so by tacitly colluding using either a double-sided collusion where they collude on both rates and commissions or a single-sided collusion where they only collude on the commission and are free to set the rates as they see fit. In general, it can be argued that single-sided collusions are easier to maintain as the platforms are then free to adopt a dynamic pricing or surge pricing strategy based on the spatio-temporal characteristics of any request. Another argument in support for a possible single-sided collusion between the platforms comes from the fact that colluding doubly on both the rates and the commission is a clear signal for any observing anti-trust authority as opposed to just colluding on a single side of the market. Finally, all of our results are derived under the assumption that drivers are optimistic optimizers i.e. under indifferent payoff, they still choose to participate on the ridesharing economy. As we argued, this is a reasonable assumption provided that there is some extra incentive (tips) involved or the alternative is to not work at all or an accurate estimate of net profits cannot be reliably made.

\subsection{Managerial Implication}
Our work has strong managerial implication for the policy makers i.e. the existence of a minimum commission for the drivers. Since we showed that it's a natural outcome for the platforms to collude together and exploit the drivers, the existence of a clear minimum wage seems to be the only measure that can adequately protect drivers' interest in a duopoly market. This is indeed in alignment with the recent ongoing efforts\footnote{https://www.startribune.com/new-battle-lines-in-minneapolis-debate-over-uber-and-lyft-
driver-pay-same-as-the-old-lines/600346563/} for establishing clear minimum wage rates in the US. Finally, we stress that this study does not accuse any existing real world platforms of colluding but merely provides a conclusion for a stylized model.

\bibliographystyle{cas-model2-names}

\bibliography{cas-refs}

\appendix


\end{document}